\documentclass[10pt, twocolumn, twoside]{IEEEtran}

\usepackage{multirow}
\usepackage{makecell}
\usepackage{hyperref}

\usepackage{pdfpages}
\usepackage{textcomp}
\usepackage{epsfig,latexsym}
\usepackage{float}
\usepackage{indentfirst}
\usepackage{amsmath}
\usepackage{amssymb}
\usepackage{amsthm}
\usepackage{xcolor}

\usepackage{times}
\usepackage{subfigure}
\usepackage{psfrag}
\usepackage{cite}
\usepackage{flushend}  
\usepackage{epstopdf}
\usepackage{fancyhdr}
\usepackage{stfloats}
\usepackage{color}
\usepackage[noend]{algpseudocode}
\usepackage{algorithmicx,algorithm}
\usepackage{graphicx}

\def\diag{\textrm{diag}}

 \newcommand{\xqedhere}[2]{%
\rlap{\hbox to#1{\hfil\llap{\ensuremath{#2}}}}}

\newtheorem{theorem}{\bf{Theorem}}

\newtheorem{remark}{\bf{Remark}}

\begin{document}%
\title{A Gradient Meta-Learning Joint Optimization for Beamforming and Antenna Position in Pinching-Antenna Systems}

\author{
\normalsize{\IEEEauthorblockN{Kang Zhou, Weixi Zhou, Donghong Cai, Xianfu Lei,~\IEEEmembership{Member,~IEEE}, Yanqing Xu,\\ Zhiguo Ding,~\IEEEmembership{Fellow,~IEEE},  and Pingzhi Fan,~\IEEEmembership{Life Fellow,~IEEE}
}}

\thanks{K. Zhou and W. Zhou are with the College of Computer Science, Sichuan Normal University, Chengdu 610031, China (e-mail: zhoukang.sicnu@foxmail.com; zhouweixi@sicnu.edu.cn). (Corresponding author: Weixi Zhou.)}
\thanks{ D. Cai is with the College of Cyber Security, Jinan University, Guangzhou 510632, China (e-mail: dhcai@jnu.edu.cn).}
\thanks{X. Lei is with the School of Information Science and Technology, Southwest Jiaotong University, Chengdu 610031, China (e-mail:
xflei@swjtu.edu.cn).}
\thanks{Y. Xu is with the School of Science and Engineering, The Chinese University of Hong Kong, Shenzhen 518172, China (e-mail: xuyanqing@cuhk.edu.cn).}
\thanks{Z. Ding is with the Department of Computer and Information Engineering, Khalifa University, Abu Dhabi, UAE. (e-mail: zhiguo.ding@manchester.ac.uk).}
\thanks{P. Fan is with the Key Lab
of Information Coding and Transmission, Southwest Jiaotong University,
Chengdu 610031, China (e-mail: pzfan@swjtu.edu.cn).}

%

}
 \vspace{1.5em}

\maketitle
\vspace{-2.0em}

\begin{abstract}

In this paper, we consider a novel optimization design for multi-waveguide pinching-antenna systems, aiming to maximize the weighted sum rate (WSR) by jointly optimizing beamforming coefficients and antenna position. To handle the formulated non-convex problem, a gradient-based meta-learning joint optimization (GML-JO) algorithm is proposed. Specifically, the original problem is initially decomposed into two sub-problems of beamforming optimization and antenna position optimization through equivalent substitution. Then, the convex approximation methods are used to deal with the nonconvex constraints of sub-problems, and two sub-neural networks are constructed to calculate the sub-problems separately. Different from alternating optimization (AO), where two sub-problems are solved alternately and the solutions are influenced by the initial values, two sub-neural networks of proposed GML-JO with fixed channel coefficients are considered as local sub-tasks and the computation results are used to calculate the loss function of joint optimization. Finally, the parameters of sub-networks are updated using the average loss function over different sub-tasks and the solution that is robust to the initial value is obtained. Simulation results demonstrate that the proposed GML-JO algorithm achieves 5.6 bits/s/Hz WSR within 100 iterations, yielding a 32.7\% performance enhancement over conventional AO with substantially reduced computational complexity. Moreover, the proposed GML-JO algorithm is robust to different choices of initialization and yields better performance compared with the existing optimization methods.

\end{abstract}
 \vspace{-0.5em}
\begin{IEEEkeywords}
Pinching-antenna, weighted sum rate, joint optimization, gradient-based meta-learning.
\end{IEEEkeywords}

 \vspace{-0.5em}
\section{Introduction}
\vspace{0.5em}

With the accelerated evolution of wireless communication technologies from 5G to 6G, emerging scenarios such as ultra-dense networks, massive machine-type communication (mMTC) \cite{01}, and ultra-reliable low-latency communication (URLLC) \cite{02} require higher spectral efficiency, more connections and higher quality of service (QoS) \cite{03,04,05}. One of the key enabling technologies of 5G and 6G is massive multiple-input multiple-output (MIMO) \cite{06,07,0008} which deployed a few hundred antennas to significantly increase system capacity and spectral efficiency. A uniform planar array antenna structure is considered in  frequency division duplexing massive MIMO systems \cite{08} because it can deploy a large number of antennas within a relatively small area. Moreover, extremely large-scale MIMO (XL-MIMO) improves the system performance by deploying extremely large number of antennas in a compact space\cite{09}. It is important to point out that the complexity of signal processing is extremely high due to the large-scale antenna array. A distributed XL-MIMO with local and central processing units is adopted to reduce the computation overhead \cite{100}, which can approach the performance of centralized antenna deployment. However, distributed XL-MIMO achieve spatial flexibility through the static combination of a large number of antennas.

Recently, flexible antenna systems, including fluid-antenna \cite{11}, movable-antenna \cite{12}, and pinching-antenna \cite{13,0015} have received extensive attention. Fluid antenna utilizes the plasticity and controllability of fluids, such as liquid metals or conductive fluids, i.e., we can be adjusted the shape and position of the fluid to mitigate channel fading\cite{11}. Similarly, movable antenna optimizes radiation direction and coverage through physical repositioning \cite{12}. However, the movement ranges of fluid and movable antennas are limited within a few wavelengths\cite{0012}. Furthermore, the deployment of fluid antennas and removable antennas is costly and lack of the antenna array reconfigurability. Pinching-antenna systems has emerged as a novel solution of flexible antennas, leveraging the dynamic adjustment of dielectric particle positions on waveguides to reconfigure electromagnetic field distributions\cite{0013,0014}. In \cite{0015}, pinching-antenna systems was explored to multi-user environments, where multiple users can be better served with established scalable strong line-of-sight (LoS) links via antenna position adjustment \cite{0013}. In\cite{13}, a beamforming design of pinching-antenna systems was considered to enhance signal directionality, which was achieved by dynamically adjusting the positions of pinching-antenna. Meanwhile, \cite{0016} revealed that LoS blocking in multi-user environments may be beneficial, which utilizes obstacles to block interfering signals through waveguide layouts, thereby suppressing co-channel interference.

From the optimization perspective, Wang et al. considered the continuous activation and discrete activation modes of pinching antennas, employing a penalty-based alternating optimization (AO) algorithm to address the problem of transmitting power minimization\cite{15}. Xu et al. \cite{16} employed a two-stage optimization strategy to maximize sum rates, i.e., the first stage optimizes the antenna position to minimize large-scale path loss, and the second stage refines the position to maximize the received signal strength. \cite{17} proposed a branch-and-bound algorithm and a many-to-many matching algorithm to minimize transmission power by jointly optimizing transmit beamforming, pinching beamforming, and the number of activated pinching-antennas. Particularly, \cite{19} employed a joint beamforming design with dual time scales, which decouples the total rate maximization problem into two subproblems using the dual decomposition method and employed the Karush-Kuhn-Tucker conditions and stochastic continuous convex approximation methods to address non-convex subproblems. Xu et al. \cite{520} designed a non-orthogonal multiple access system with pinching-antennas, where AO and successive convex approximation were used to joint optimize the transmission data under QoS requirements. Furthermore, pinching-antenna systems have been applied in various communication systems. For example, \cite{20} implemented pinching-antenna systems in over-the-air computation systems, where a mean squared error minimization problem was sovled by jointly optimizing the pinching-antenna position, transmit power, and decoding vector. Furthermore, deep learning algorithms have been applied in pinching-antenna systems. Guo et al. \cite{21} pioneered the use of graph neural networks (GNNs) for end-to-end beamforming learning, which realized a joint learning of pinching beamforming and transmission beamforming. Addressing channel estimation challenges of pinching-antenna systems, \cite{22} proposed pinching-antenna mixture of experts and pinching-antenna transformer models that leverage mixture-of-experts and self-attention mechanisms, surpassing traditional optimisation methods in scalability and zero-shot learning capability. Zhou et al. \cite{521} proposed a novel channel estimation for mmWave pinching antenna systems, where only a few pinching antennas are activates to reduce antenna switching and pilot overhead costs,  thereby achieving high precision estimation of multipath propagation parameters.

In this paper, we consider a multi-waveguide multi-user pinching-antenna systems and formulate an optimization problem for maximizing the weighted sum rate (WSR). The formulated maximizing WSR problem is non-convex and is challenge to be solved due to the coupling between the beamforming and antenna positions. To handle this problem, we propose a gradient-based meta-learning joint optimization (GML-JO) algorithm. Unlike the AO methods, we first decompose the objective function of original problem into two subproblems including beamforming optimization and antenna position optimization by Lagrangian duality and quadratic transforms. Then the non-convex subproblems are transformed into tractable convex forms using convex approximation techniques, and two sub-networks are employed to handle the sub-tasks. Finally, the sub-problem with one fixed channel is considered as a computation sub-task and the the average loss is computed over multiple sub-tasks to obtain the average loss for training the network’s model parameters. The contributions of this paper are summarized as follows:
\begin{itemize}
\item  An optimization problem is formulated to maximize the WSR of a multi-waveguide multi-user pinching-antenna systems. Through Lagrangian duality and quadratic transformations, we perform an equivalent transformation of the original objective function and decompose it into two subproblems, i.e., beamforming optimization and antenna position optimization. Additionally, using convex approximation techniques, we transform the non-convex sub-problems into solvable convex forms.

\item  Furthermore, we propose a GML-JO algorithm to solve the problem. Especially, we design an unfolded neural network to efficiently compute the solutions alternately for each sub-problem. Different from AO, where the obtained feasible solution is impacted by the initial value, solving each sub-problem of our proposed GML-JO with a fixed channel is first regarded as the calculation of a sub-task, and then a solution that is robust to the initial value can be obtained.

\item The designed unfolded neural network of proposed GML-JO is used to calculate each sub-task alternately. And the objective function of original problem with coupling variables is defined as loss function for joint optimization. With the computation results from multiple sub-tasks, the average loss is obtained and used to update the model parameters of two sub-network, thereby improving the overall performance of the solutions for sub-problems.

\end{itemize}

Experimental results demonstrate that the proposed GML-JO algorithm achieves rapid convergence and 5.6 bits/s/Hz WSR is obtained within 100 iterations, delivering a 32.7\% performance enhancement compared to the benchmark AO method. Moreover, the solution of proposed GML-JO is more robust with respect to the initial values.

The paper is organized as follows: Section II introduces the system model and problem formulation. Section III describes the equivalent transformation of the original problem and its AO processing. Section IV presents the proposed GML-JO algorithm. Section V presents the numerical results, and Section VI concludes this paper finally.

\textit{Notations:} $a, \mathbf{a}, \mathbf{A}$ denote scalar, vector and matrix, respectively. $\diag(\mathbf{a})$ is the diagonal
matrix with diagonal elements $\mathbf{a}$. $\mathcal{CN}(\mu,\sigma^2)$ denotes complex Gaussian random variable with mean $\mu$ and variance $\sigma^2$.

\section{System Model}
\vspace{1.0em}

\begin{figure}[tp]
  \centering
\includegraphics[width=0.5\textwidth]{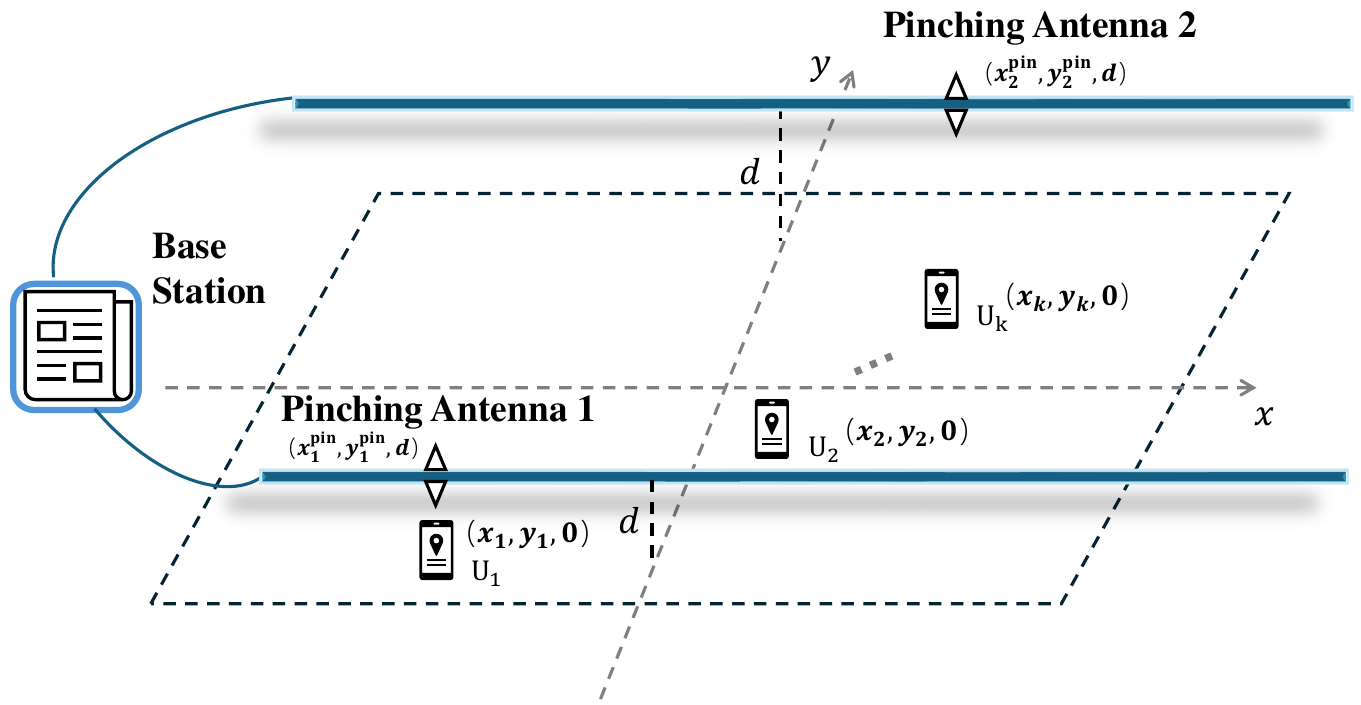}
  \caption{Downlink multi-user multi-waveguide pinching-antenna systems.}
  \label{fig:example}
\end{figure}

As shown in Fig. 1, we consider a multi-waveguide pinching-antenna communication system, where \( K \) equal distance distributed waveguides serves $M$ single-antenna users. Without loss of generality, only a single pinching-antenna is activated on each waveguide. The position of the pinching antenna on the \( k \)-th $(k=1,2,\cdots,K)$ waveguide is denoted by \( \tilde{\psi}_k^{\text{Pin}}  =  (x_k, y_k)\), which is a two-dimensional coordinate in the predetermined service region \( S \).  The base station supports multi-user communications by passing signals through waveguides to pinching antennas. The received signal \( y_m \) at the \( m \)-th user is given by
\begin{equation}
{y_m} = \sum_{k=1}^{K} h_{m,k} p_{m,k}\sqrt{P}  {s_m}+ \sum_{k=1}^{K} \sum_{i \neq m} h_{m,k} p_{i,k} \sqrt{P}{s_i}+ n,
\end{equation}
where \( h_{m,k} \) denotes the channel gain between the \( m \)-th user and the pinching-antenna on the \( k \)-th waveguide, \( p_{m,k} \) is the the beamforming coefficient \cite{0015} containing both the amplitude and phase
information for the desired signal \( s_m \) of the \( m \)-th user from the $k$-th waveguide with $\mathbb{E}[s^Hs]=1$, and \( n \sim \mathcal{CN}(0, \sigma^2) \) is the additive white Gaussian noise. Especially, the channel gain \( h_{{m,k}} \) combines the multiple effects of large-scale path loss and phase variation due to in-waveguide propagation, defined as
\begin{equation}\label{lsod}
h_{m,k}
= \frac{
\sqrt{\eta \,} e^{-2\pi j \Bigl(\frac{1}{\lambda}\,\bigl|\psi_m - \tilde{\psi}_k^{\mathrm{Pin}}\bigr|
+ \frac{1}{\lambda_g}\,\bigl|\psi_{0,k}^{\mathrm{Pin}} - \tilde{\psi}_k^{\mathrm{Pin}}\bigr|\Bigr)}}
{
\bigl|\psi_m - \tilde{\psi}_k^{\mathrm{Pin}}\bigr|
},
\end{equation}
where \( \eta \) is the attenuation factor, \( \lambda \) and \( \lambda_g \) denote the wavelength of the signal in free space and waveguide, respectively; \( \psi_m \) and \( \psi_k^{\text{Pin}} \) denote the positions of the pinching antenna on the \( m \)-th user and \( k \)-th waveguide, respectively; and \( \psi_{0,k}^{\text{Pin}} \) is the reference position of the \( k \)-th pinching-antenna. Then the signal-to-interference-plus-noise ratio (SINR) for the \( m \)-th user can be defined as
\begin{equation}\label{sldor}
\text{SINR}_m = \frac{{P}\left| \sum_{k=1}^{K} h_{m,k} p_{m,k} \right|^2}{ \sum_{i \neq m} {P}\left| \sum_{k=1}^{K} h_{m,k} p_{i,k} \right|^2 + \sigma^2},
\end{equation}
Let
$\mathbf{p}_m = \begin{bmatrix} p_{m,1}, \cdots, p_{m,K} \end{bmatrix}^T,
\mathbf{h}_m = \begin{bmatrix} h_{m,1}, \cdots, h_{m,K} \end{bmatrix}^T,$
the SINR in \eqref{sldor} can be rewritten as
\begin{equation}
\text{SINR}_m
= \frac{ {P}\bigl|\mathbf{h}_m^H \mathbf{p}_m\bigr|^2}
       { {P}\sum_{i \neq m} \bigl|\mathbf{h}_m^H \mathbf{p}_i\bigr|^2 + \sigma^2}.
\end{equation}

It is important to point out that the exponential term in \eqref{lsod} includes both the phase shifts of the signal propagating in free space and inside the waveguide.
A WSR optimization problem can be formalized as
\begin{equation}
\max_{\mathbf{d}_k,\mathbf p_m} \sum_{m=1}^{M} {w_m} \log_2 \left( 1 + \text{SINR}_m \right),
\end{equation}
where \( \omega_m \in [0,1] \) denotes the fairness weights assigned to the $m$-th user, and \(\mathbf d_k \) denotes the position vector of the \( k \)-th antenna. Assuming that the waveguide is fixed along the y-axis and the antenna positions are adjusted only along the x-axis. Define \(\mathbf d = [x_1, x_2,...,x_k]^T \) with the position of the \( k \)-th antenna \( x_k \), and the optimization problem can be rewritten as
\begin{subequations}\label{mosd2}
\begin{align}
& \max_{\mathbf{d}, \mathbf{p}_m} \sum_{m=1}^{M} w_m \log_2 \left( 1 + \frac{ {P}\left| \mathbf{h}_m^H\mathbf{p}_m \right|^2}{ \sum_{n \neq m} {P}\left| \mathbf{h}_m^H\mathbf{p}_n \right|^2 + \sigma^2} \right) \\
&\mathrm{s.t. }\quad \sum_{m=1}^{M} \|\mathbf{p}_m\|^2 \leq P_{\text{total}}, \label{sldo61}\\
&\quad\quad \quad {d_k}\in D , \quad \forall k = 1, 2, \dots, K, \label{sld32}\\
& \quad\quad\quad \text{SINR}_m \geq \gamma_m^{\min}, \quad \forall m = 1, 2, \dots, M, \label{stre71}
\end{align}
\end{subequations}
where \( \gamma_m^{\text{min}} \) denotes the user's minimum quality of service;
constraint \eqref{sldo61} denotes the total transmit power limit; constraint \eqref{sld32} requires the position of each pinching-antenna in a predetermined region \( S \), and constraint \eqref{stre71} is quality of service condition. The optimization problem \eqref{mosd2} is challenging due to its non-convexity and the strong coupling between the optimization parameters.

\section{Alternate Optimization of Original Problem}
\vspace{1.0em}

In this section, we first approximate the original optimization problem by introducing an auxiliary variable to turn it, and then combine it with the quadratic transformation theorem to reconstruct the original problem into an easy-to-handle equivalent optimization form, and further solve the optimization problem by using the AO method.

\subsection{Equivalent Transformation of Original Problem \eqref{mosd2}}

To deal with the nonconvexity of the original WSR maximization problem \eqref{mosd2}, we adopt the Lagrangian dyadic transformation method to transform problem \eqref{mosd2} into an equivalent optimization problem by introducing auxiliary variables \cite{043}. In particular, with the help of the introduced auxiliary variable, the rate $\log_2(1+\mathrm{SINR}_m)$ in the objective function of problem \eqref{mosd2} can be expressed as the optimization problem in \textbf{Theorem 1}.
\begin{theorem}
With the $\text{SINR}_m$ of user $m$, the rate $\log_2(1+\mathrm{SINR}_m)$ can  be formulated as
\begin{align}\label{sldo2}
\log_2(1 + \text{SINR}_m) &= \max_{\gamma_m \geq 0} \left[ \log_2(1 + \gamma_m) - \frac{\gamma_m}{\ln 2} \right. \nonumber\\
&\quad + \left. \frac{(1 + \gamma_m) \text{SINR}_m}{\ln 2 (1 + \text{SINR}_m)} \right].
\end{align}
\end{theorem}

\begin{proof}
We first define an auxiliary function as
\begin{equation}
f(\gamma_m) = \log_2(1 + \gamma_m) - \frac{\gamma_m}{\ln 2} + \frac{(1 + \gamma_m) \text{SINR}_m}{\ln 2(1 + \text{SINR}_m)}.
\end{equation}

Then the derivative of $f(\gamma_m)$ with respect to the variable $\gamma_m$ is given by
\begin{equation}
\frac{\partial f}{\partial \gamma_m} = \frac{1}{\ln 2} \cdot \frac{1}{1 + \gamma_m} - \frac{1}{\ln 2} + \frac{\text{SINR}_m}{\ln 2 (1 + \text{SINR}_m)}.
\end{equation}
Let the derivative be zero, we have
\begin{equation}\label{lspd}
\frac{1}{1 + \gamma_m} - 1 + \frac{\text{SINR}_m}{1 + \text{SINR}_m} = 0.
\end{equation}
By further simplifying \eqref{lspd}, results in
\begin{equation}
\frac{1}{1 + \gamma_m} = 1 - \frac{\text{SINR}_m}{1 + \text{SINR}_m} = \frac{1}{1 + \text{SINR}_m}.
\end{equation}
Furthermore, the second-order derivative of $f(\gamma_m)$ with respect to the variable $\gamma_m$ is $\frac{\partial^2 f}{\partial \gamma_m^2}=-\frac{1}{(1 + \gamma_m)^2}<0$. Thus, the optimal solution of problem \eqref{sldo2} can be expressed as
\begin{equation}\label{sldow21}
\gamma_m^* = \text{SINR}_m.
\end{equation}
Substituting the optimal value $\gamma_m^*$ into $f(\gamma_m)$, we have
\begin{align}
f(\text{SINR}_m) &= \log_2(1 + \text{SINR}_m) - \frac{\text{SINR}_m}{\ln 2} \nonumber\\
&\quad + \frac{(1 + \text{SINR}_m) \text{SINR}_m}{\ln 2 (1 + \text{SINR}_m)} \nonumber\\
&= \log_2(1 + \text{SINR}_m) - \frac{\text{SINR}_m}{\ln 2} + \frac{\text{SINR}_m}{\ln 2} \nonumber\\
&= \log_2(1 + \text{SINR}_m).
\end{align}
\end{proof}
\begin{remark}
By using the theorem, the data rate $\log_2(1+\mathrm{SINR}_m)$ which is the objective function in problem \eqref{mosd2} can be equivalently expressed as optimization problem \eqref{sldo2} by introducing variable $\gamma_m$. Especially, $\log_2(1+\mathrm{SINR}_m)$ can be obtained at the optimal $\gamma_m^{*}=\mathrm{SINR}_m$.
\end{remark}

However, the fractional term $\frac{\text{SINR}_m}{1 + \text{SINR}_m}$ in \eqref{sldo2} is still difficult to optimize directly. By introducing another auxiliary variable, fractional term $\frac{\text{SINR}_m}{1 + \text{SINR}_m}$ can be further recast with quadratic transformation \cite{044} shown in following theorem.
\begin{theorem}
For any $\text{SINR}_m \geq 0$, the fractional term $\frac{\text{SINR}_m}{1 + \text{SINR}_m}$ in \eqref{sldo2} can be formulated as
\begin{equation}
\begin{aligned}
\frac{\text{SINR}_m}{1 + \text{SINR}_m} &= \max_{y_m \geq 0} \left[ 2 y_m \sqrt{G_m} - y_m^2 (G_m + I_m + \sigma^2) \right].
\end{aligned}
\end{equation}
\end{theorem}

\begin{proof}
Let $G_m = |\mathbf{h}_m^H \mathbf{p}_m|^2$ and $I_m = \sum_{n \ne m} |\mathbf{h}_m^H \mathbf{p}_n|^2$, then fractional term $\frac{\text{SINR}_m}{1 + \text{SINR}_m}$ can be expressed as
\begin{equation}
\frac{\text{SINR}_m}{1 + \text{SINR}_m} = \frac{G_m}{G_m + I_m + \sigma^2}.
\end{equation}
According to the quadratic transformation \cite{044}, we define a function $f(y_m)$ as
\begin{equation}
f(y_m) = 2 y_m \sqrt{G_m} - y_m^2 (G_m + I_m + \sigma^2).
\end{equation}
The derivative of $f(y_m)$ can be obtained by
\begin{equation}
f'(y_m) = 2 \sqrt{G_m} - 2 y_m (G_m + I_m + \sigma^2).
\end{equation}
Let the derivative be zero, i.e.,
\begin{equation}
2 \sqrt{G_m} - 2 y_m (G_m + I_m + \sigma^2) = 0,
\end{equation}
then, we have
\begin{equation}\label{tr43}
y_m^* = \frac{\sqrt{G_m}}{G_m + I_m + \sigma^2}.
\end{equation}
Since $G_m > 0$ and $G_m+I_m + \sigma^2 > 0$, thus the condition $y_m > 0$ is satisfied. To verify that $y_m^*$ is the maximum point, the second-order derivative is calculated by

\begin{equation}
f''(y_m) = \frac{d}{d y_m} \left( 2 \sqrt{G_m} - 2 y_m (G_m + I_m + \sigma^2) \right),
\end{equation}
and
\begin{equation}
f''(y_m) = -2 (G_m + I_m + \sigma^2) < 0.
\end{equation}
By substituting $y_m^*$ into $g(y_m)$, we have
\begin{align}
g(y_m^*) &= 2 \left( \frac{\sqrt{G_m}}{G_m + I_m + \sigma^2} \right) \sqrt{G_m} \nonumber\\
&\quad - \left( \frac{\sqrt{G_m}}{G_m + I_m + \sigma^2} \right)^2 (G_m + I_m + \sigma^2)\nonumber\\
&= 2 \cdot \frac{G_m}{G_m + I_m + \sigma^2} - \frac{G_m}{G_m + I_m + \sigma^2} \nonumber\\
          &= \frac{G_m}{G_m + I_m + \sigma^2} \nonumber\\
          &= \frac{\text{SINR}_m}{1 + \text{SINR}_m}.
\end{align}
\end{proof}
Based on Theorem 1 and Theorem 2, the original problem \eqref{mosd2} is converted to
\begin{subequations}\label{mosd5}
\begin{align}
& \max_{\mathbf{d}, \mathbf{p}_m, \gamma_m, y_m} \sum_{m=1}^{M} w_m \left[ \log_2(1 + \gamma_m) - \frac{\gamma_m}{\ln 2} + \frac{1 + \gamma_m}{\ln 2} \right. \nonumber\\
&\quad\quad\quad \left( 2 y_m \sqrt{G_m} - y_m^2 (G_m + I_m + \sigma^2) \right), \label{sldp21} \\
&\mathrm{s.t.} \quad \sum_{m=1}^{M} \|\mathbf{p}_m\|^2 \leq P_{\text{total}}, \label{sldo14} \\
&\quad\quad \quad d_k \in D, \quad \forall k = 1, 2, \dots, K, \label{sld14} \\
& \quad\quad\quad \text{SINR}_m \geq \gamma_m^{\min}, \quad \forall m = 1, 2, \dots, M, \label{stre14}\\
&\quad\quad\quad\gamma_m \geq 0, \quad y_m \geq 0, \quad \forall m = 1, 2, \dots, M. \label{stre15}
\end{align}
\end{subequations}

\textbf{Theorem 1} and \textbf{Theorem 2} provide an equivalence transformation of the original optimization problem \eqref{mosd2} by introducing variables $\gamma_m$ and $y_m$. Note that the objective function \eqref{sldp21} can be easier handled compared to the objective function \eqref{sldo61}. However, problem \eqref{mosd5} is still difficult to deal with because of the constraints on variables $d_k$ and $\mathbf{p}_m$.

\subsection{Alternate Optimization of Problem \eqref{mosd5}}

It is important to point out that problem \eqref{mosd5} is still nonconvex due to the coupling of variables $\mathbf{p}_m$, $\mathbf{d}$, $\gamma_m$, and $y_m$. Thus, we solve problem \eqref{mosd5} alternately, including the auxiliary variables update and subproblems update.
\subsubsection{Auxiliary variables update}
The auxiliary variables $\gamma_m$ and $y_m$ can be obtained by \eqref{sldow21} and \eqref{tr43} with given $G_m$ and $I_m$. Thus, the auxiliary variables update are given by
\begin{equation}\label{mosd520}
\gamma_m^{(t+1)} = \frac{G_m^{(t)}}{I_m^{(t)} + \sigma^2},
\end{equation}
and
\begin{equation}\label{mosd521}
y_m^{(t+1)} = \frac{\sqrt{G^{(t)}_m}}{G^{(t)}_m + I^{(t)}_m + \sigma^2},
\end{equation}
where $G^{(t)}_m$ and $I^{(t)}_m$ denote the $t$-th iteration of $G_m$ and $I_m$.

\subsubsection{{Beamforming Coefficient $\mathbf{p}_m$} update}
For given $y_m^{(t)}$, $\mathbf{d}^{(t)}$ and $\gamma_m^{(t)}$, the subproblem of problem \eqref{mosd5} on variable $\mathbf{p}_m$ can be formulated as
\begin{subequations}\label{sldp32}
\begin{align}
& \max_{\mathbf{p}_m} \sum_{m=1}^{M} \frac{w_m (1 + \gamma_m^{(t)})}{\ln 2} \bigg[ 2y_m^{(t)} \sqrt{G_m} - \left(y_m^{(t)}\right)^2 \nonumber\\
&\quad\left(G_m + I_m + \sigma^2\right) \bigg],\label{obj2}\\
& \mathrm{s.t.} \sum_{m=1}^{M} \|\mathbf{p}_m\|^2 \leq P_{\text{total}}, \label{sldo611} \\
& \quad\mathrm{SINR}_m = \frac{G_m}{\sigma^2 + I_m} \geq \gamma_m^{\min}, \quad \forall m = 1, 2, \dots, M. \label{stre712}
\end{align}
\end{subequations}

Since $G_m$ and $I_m$ are nonconvex functions with respect to variable $\mathbf{p}_m$, the objective function and the $\mathrm{SINR}_m$ constraint in subproblem \eqref{sldp32} are nonconvex. To solve this problem, we perform a first-order Taylor expansion of $G_m$ at $\mathbf{p}_m^{(t)}$, i.e.,
\begin{equation}\label{mosd522}
\sqrt{G_m} \approx \sqrt{G_m^{(t)}} + \frac{1}{2\sqrt{G_m^{(t)}}} \text{Re} \left( \mathbf{h}_m^H ( \mathbf{p}_m - \mathbf{p}_m^{(t)} ) \right).
\end{equation}
Similarly, the linearization$ ofI_m$ can be approximated as
\begin{equation}\label{mosd523}
I_m \approx I_m^{(t)} + 2\!\sum_{n \neq m} \text{Re} \left( \left( \mathbf{p}_n^{(t)} \right)^H \!\mathbf{h}_m \mathbf{h}_m^H \left( \mathbf{p}_n - \mathbf{p}_n^{(t)} \right) \right).
\end{equation}
Then the objective function \eqref{obj2} can be approximated by
\begin{align}
&\tilde{f}_p(\mathbf{p}_m ) \triangleq \sum_{m=1}^{M} w_m \frac{1 + \gamma_m^{(t)}}{\ln 2} \Bigg[ 2y_m^{(t)} \bigg( \sqrt{G_m^{(t)}} + \frac{1}{2\sqrt{G_m^{(t)}}}\nonumber\\ &\cdot\text{Re} \left( h_m^H (\mathbf{p}_m - \mathbf{p}_m^{(t)}) \right) \bigg)- \left( y_m^{(t)} \right)^2 \left( G_m^{(t)} + I_m^{(t)} \right)\nonumber \\
&+ 2 \sum_{n \neq m} \text{Re} \left( \left( \mathbf{p}_n^{(t)} \right)^H h_m h_m^H \left( \mathbf{p}_n - \mathbf{p}_n^{(t)} \right) \right) + \sigma^2 \Bigg].
\end{align}
The $\mathrm{SINR}_m$ constraint \eqref{stre712} can be approximated by
\begin{equation}
\begin{aligned}\label{sinre2}
G_m - \gamma_{\min}^{m} \Big( & I_m^{(t)} + 2 \sum_{n \neq m} \text{Re} \left( \left( \mathbf{p}_n^{(t)} \right)^H \mathbf{h}_m \mathbf{h}_m^H \left( \mathbf{p}_n - \mathbf{p}_n^{(t)} \right) \right) \\
& \quad \quad + \sigma^2 \Big) \geq 0.
\end{aligned}
\end{equation}
Therefore, the subproblem \eqref{sldp32} can be approximated by
\begin{subequations}\label{lotr}
\begin{align}
& \max_{\mathbf{p}_m} \tilde{f}_p(\mathbf{p}_m )\\
&\mathrm{s.t.} \quad\eqref{sldo611}, \eqref{sinre2}.
\end{align}
\end{subequations}

Note that problem \eqref{lotr} is a convex optimization problem, which can be easily solved by CVX. To facilitate the joint optimization of problem \eqref{mosd5} in the next section, a neural network will be designed to solve subproblem \eqref{lotr}. Specially, we obtain the following subproblem by ignoring constraint \eqref{sldo611}, i.e.,
\begin{subequations}\label{lotr22}
\begin{align}
& \max_{\mathbf{p}_m} \tilde{f}_p(\mathbf{p}_m )\\
&\mathrm{s.t.}~~~ \eqref{sinre2}.
\end{align}
\end{subequations}
Utilizing the quadratic penalty function method, constrained problem \eqref{lotr22} is transformed into an unconstrained problem, i.e.,
\begin{equation}\label{oier5}
\max_{\mathbf{p_m} }\tilde{f}_p^{\text{new}}(\mathbf{p}_m ) = \tilde{f}_p(\mathbf{p}_m ) - \mu \sum_{m=1}^{M} V_m^2(\mathbf{p}_m ),
\end{equation}
where $\mu>0$ is the penalty factor, and
\begin{equation}
V_m = \max \left( 0, \gamma_m^{\min} I_m^{\text{approx}} - G_m \right)
\end{equation}
with
\begin{equation}
I_m^{\text{app}} = I_m^{(t)} + 2 \sum_{n \neq m} \text{Re} \left( \left( \mathbf{p}_n^{(t)} \right)^H \mathbf{h}_m \mathbf{h}_m^H \left( \mathbf{p}_n - \mathbf{p}_n^{(t)} \right) \right) + \sigma^2.
\end{equation}

It is important to point out that the objective function $\tilde{f}_p^{\text{new}}(\{ \mathbf{p}_m \})$ is differentiable and its gradient is given by
\begin{equation}
\nabla_{\mathbf{p_m}} \tilde{f}_p^{\text{new}} = \nabla_{\mathbf{p_m}} \tilde{f}_p - 2\mu \sum_{m=1}^{M} V_m \cdot \nabla_{\mathbf{p_m}} V_m.
\end{equation}
Then the solution of \eqref{oier5} is obtained by
\begin{equation}\label{etl32}
\mathbf{p}_m^{(t+1)} = \mathbf{p}_m^{(t)} + \eta \nabla_{\mathbf{p_m}} \tilde{f}_p^{\text{new}},
\end{equation}
where $\eta$ is the learning rate. To satisfy the constraint \eqref{sldo611} in \eqref{lotr}, the updated $\mathbf{p_m}^{(t+1)}$ in \eqref{etl32} is then normalized by
\begin{equation}\label{rtle3}
\mathbf {p}_m^{(t+1)} = \frac{\sqrt{P_{\text{total}}}}{\sqrt{\sum_{m=1}^{M} \left| \mathbf {p}_m^{(t+1)} \right|^2}} \, \mathbf {p}_m^{(t+1)}.
\end{equation}

\subsubsection{{Antenna Position $\mathbf{d}$ update}} With fixing $y_m^{(t)}$, $\mathbf{p}_m^{(t+1)}$ and $\gamma_m^{(t)}$, the subproblem of problem \eqref{mosd5} on variable $\mathbf{d}$ can be formulated as

\begin{subequations}\label{lotr5}
\begin{align}
& \max_{\{\mathbf{d}\}} {f}(\mathbf{d}) \triangleq\sum_{m=1}^{M} \frac{w_m (1 + \gamma_m^{(t)})}{\ln 2} \bigg[ 2y_m^{(t)} \sqrt{G_m} - \left(y_m^{(t)}\right)^2 \nonumber\\
&\quad\quad\left(G_m + I_m + \sigma^2\right) \bigg],\\
&\mathrm{s.t.} \quad d_k \in D, \quad \forall k = 1, 2, \dots, K. \label{lert3}
\end{align}
\end{subequations}

To maximize $f(\mathbf{d})$, the gradient descent method is considered and the gradient of the objective function is given by
\begin{equation}\label{tre45}
\nabla_{\mathbf{d}} f(\mathbf{d}) = \left[ \frac{\partial f}{\partial d_1}, \frac{\partial f}{\partial d_2}, \dots, \frac{\partial f}{\partial d_K} \right]^T.
\end{equation}

Note that the derivation of $f(\mathbf{d})$ needs to be combined with the chain rule, and the analytical expression cannot be obtained. With the calculation of the gradient $\nabla_{\mathbf{d}} f(\mathbf{d})$, the update of $\mathbf{d}$ is given by
\begin{equation}
{\mathbf{d}}^{(t+1)} =  {\mathbf{d}}^{(t)} + \alpha \nabla_{\mathbf{d}} f \left( {\mathbf{d}}^{(t)} \right),
\end{equation}
where $\alpha$ is the step size. To ensure that $\mathbf{d}$ satisfies the constraint \eqref{lert3}, the updated $\mathbf{d}^{(t+1)}$ needs to be projected to the feasible set $D$, i.e.,
\begin{equation}\label{rty56}
{\mathbf{d}}^{(t+1)} = \text{proj}_D \left( {\mathbf{d}}^{(t)} + \alpha \nabla_{\mathbf{d}} f \left( {\mathbf{d}}^{(t)} \right) \right).
\end{equation}
\begin{remark}
The calculation and projection of the gradient in \eqref{tre45} are difficult to achieve in the traditional gradient descent algorithm. Thus, we can expand the gradient descent algorithm into a deep neural network and implement it based on Pytorch.
\end{remark}
\begin{remark}
The solution of problem \eqref{mosd5} by alternating updates of auxiliary variables update and the solutions of subproblems is only a feasible solution, which is affected by the initial value update.
\end{remark}

\section{Joint Optimization Based on Gradient Meta-Learning}
\vspace{1.0em}

In this section, we propose a gradient meta-learning joint optimization (GML-JO) algorithm to solve problem \eqref{mosd5} including auxiliary variables update, beam assignment vector update and pinching-antenna position vector update. In particular, beam assignment vector update and pinching-antenna position vector update are extended into two networks, namely, the beamforming network (BN) and the antenna position network (AN), which optimize the beamforming vector and the antenna positions, respectively. Moreover, the proposed GML-JO algorithm contains a three-layer loop structure, where the inner loop quickly updates the solutions of subproblems \eqref{lotr} and \eqref{lotr5}; the middle loop accumulates the loss function; and the outer loop dynamically adjusts the meta-parameters by the Adam algorithm, realizing the transition from the local fast response to the global robust convergence.


\subsection{ Beamforming Coefficient Optimization Network}
\begin{figure*}[t]
  \centering
\includegraphics[width=0.8\textwidth]{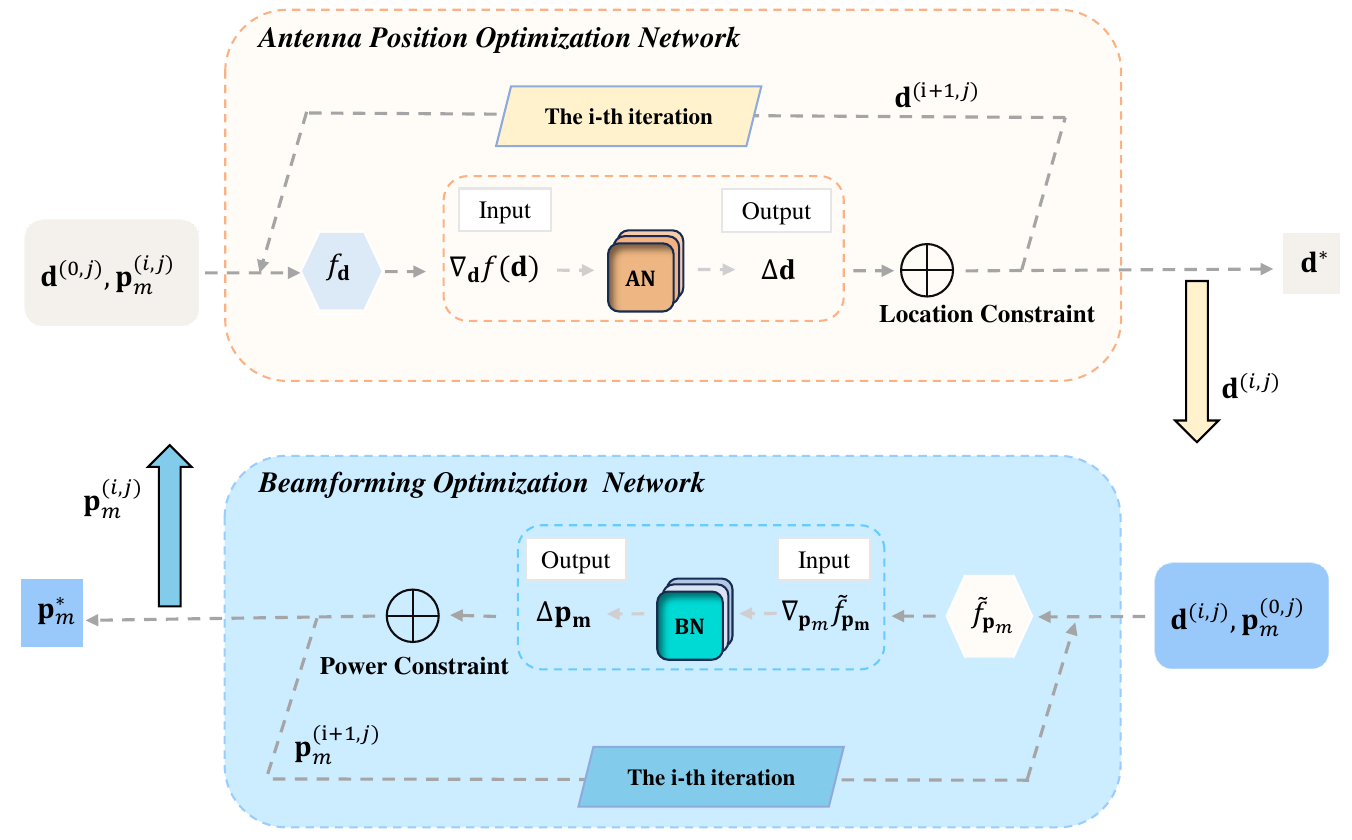}
  \caption{The sub-networks of our proposed GML-JO algorithm.}
  \label{fig:example2125}
\end{figure*}
As shown in Fig. \ref{fig:example2125}, we introduce a gradient-input unfolding network to solve subproblem \eqref{lotr} with respect to the parameter $\mathbf{p_m}$, where the input is the gradient of the objective function \eqref{oier5}. To satisfy the input of neural network, the complex gradient is first transformed into a real vector, i.e.,
\begin{equation}
\mathbf{g_p} = \left[ \text{Re} \left(\nabla_{\mathbf{p_m}} \tilde{f}_p^{\text{new}} \right)^T, \text{Im} \left( \nabla_{\mathbf{p_m}} \tilde{f}_p^{\text{new}}\right)^T \right]^T \in \mathbb{R}^{2KM}.
\end{equation}

Then three linear hidden layers are designed to learning the gradient change in \eqref{etl32}, i.e.,
\begin{equation}
h_1 = \text{ELU}(W_1 \mathbf{g_p} + b_1),
\end{equation}
and
\begin{equation}
h_2 = \text{ELU}(W_2 h_1 + b_2).
\end{equation}
The output of the network is given by
\begin{equation}
\Delta \mathbf{p_m} = \tanh(W_3 h_2 + b_3) \in [-1, 1]^{2KM},
\end{equation}
where $\mathrm{tanh}(\cdot)$ is the activation function.
Thus $\mathbf{p_m}$ updated in \eqref{etl32} is further expressed as

\begin{equation}
\mathbf {p}_m^{(i+1,j)} = \mathbf {p}_m^{(i,j)} + \Delta \mathbf {p}_m^{(i,j)}.
\end{equation}

Finally, updated $\mathbf{p_m}^{(i+1,j)}$ is then normalized by
\eqref{rtle3} to satisfy the constraint \eqref{sldo611}.

\subsection{ Antenna Position Optimization Network}

We introduce a gradient-input unfolding network with respect to $\mathbf{d}$. Especially, the last updated position gradient is the input of the network.

Then a three-layer  network is designed, i.e.,
\begin{equation}
\mathbf{h}_1 = \text{ELU}(\mathbf{W}_1 \nabla_{\mathbf{d}} f(\mathbf{d}) + \mathbf{b}_1),
\end{equation}
\begin{equation}
\mathbf{h}_2 = \text{ELU}(\mathbf{W}_2 \mathbf{h}_1 + \mathbf{b}_2),
\end{equation}
and the output is obtained by
\begin{equation}
\Delta \mathbf{d} = \tanh(\mathbf{W}_3 \mathbf{h}_2 + \mathbf{b}_3) \in [-1, 1]^K.
\end{equation}
Then, $\mathbf{d}$ updated in \eqref{rty56} is further expressed as
\begin{equation}
\mathbf{d}^{(i+1,j)} = \mathbf{d}^{(i,j)} + \Delta \mathbf{d}^{(i,j)}.
\end{equation}
According to \eqref{rty56}, the updated $\mathbf{d}^{(i+1,j)}$ is projected into $\mathbf{D}$, given by
\begin{equation}
\mathbf{d}^{(i+1,j)} = \text{proj}_D \left( \mathbf{d}^{(i+1,j)} \right),
\end{equation}
where $\text{Proj}_\mathbf{D}(\cdot)$ is defined as

\begin{equation}
\text{Proj}_\mathbf{D}(x) =
\begin{cases}
x_{\text{min}} & \text{if } x < x_{\text{min}}, \\
x_{\text{max}} & \text{if } x > x_{\text{max}}, \\
x & \text{otherwise}.
\end{cases}
\end{equation}


\subsection{Proposed GML-JO for Problem \eqref{mosd5}}

With the expansion of beamforming optimization and antenna position optimization into the complex domain network BN and real domain network AN, respectively, we propose a GML-JO algorithm for problem \eqref{mosd5}. Especially, the solution of the allocation of problem \eqref{mosd5} is also determined by the channel coefficients, and we consider the solving of problem \eqref{mosd5} as one sub-task of the proposed GML-JO for each channel coefficient. Furthermore, the parameters of BN and AN are updated by considering the average loss functions of multiple sub-tasks.

\subsubsection{Sub-Task Computing}
For each sub-task, the update is responsible for cyclically optimizing the variables $\mathbf{p_m}$ and $\mathbf{d}$. As shown in Fig. \ref{fig:example2125}, the BN and AN are used to compute the solution of sub-problems on $\mathbf{p_m}$ and $\mathbf{d}$, respectively. In each update, we sequentially update $\mathbf{p_m}$ and $\mathbf{d}$, which can be expressed as
\begin{equation}\label{lot342}
\mathbf {p}_m^{(i+1,j)} = \text{BN}_{\theta_{\mathbf{p}}}\left(\mathbf{d}^{(i,j)}, \mathbf {p}_m^{(0,j)} \right),
\end{equation}
and
\begin{equation}\label{lort54}
\mathbf{d}^{(i+1,j)} = \text{AN}_{\theta_{\mathbf{d}}}\left(\mathbf{p}_m^{(i,j)}, \mathbf{d}^{(0,j)} \right),
\end{equation}
where $\mathbf{p_m}^{(i,j)}$ and $\mathbf{d}^{(i,j)}$ represent the $\mathbf{p_m}$ and $\mathbf{d}$ in the $i$-th update of sub-network for the $j$-th sub-task.

\subsubsection{Loss Function}
For $N_i$ sub-tasks, the loss function of each sub-task is given by
\begin{equation}
L^{(j)} = - \sum_{m=1}^{M} \omega_m \log_2 \left( 1 + \text{SINR}_m^{(j)} \right),
\end{equation}
and the average loss function is expressed as
\begin{equation}\label{rot43}
L = \frac{1}{N_j} \sum_{j=1}^{N_j} L^{(j)}.
\end{equation}
It is important to point out that the average loss \eqref{rot43} can used to adjust the parameters of sub-networks in \eqref{lot342} and \eqref{lort54}.

\subsubsection{Sub-Network Parameter Update}

With the average loss, the sub-network parameters are updated by using gradient backpropagation and the Adam optimization algorithm, i.e.,
\begin{equation}
\theta_\mathbf {p}^* = \theta_\mathbf {p} + \alpha_\mathbf {p} \cdot \text{Adam}(\nabla_{\theta_\mathbf {p}}, L, \theta_\mathbf {p}).
\end{equation}
and
\begin{equation}
\theta_\mathbf{d}^* = \theta_\mathbf{d} + \alpha_\mathbf{d} \cdot \text{Adam}(\nabla_{\theta_\mathbf{d}}, L, \theta_\mathbf{d}).
\end{equation}
The sub-network parameter are updated by using the average loss, which ensures that the whole optimization process does not fall into a local optimum. The number of neurons in each layer of network is listed in Table I, and the details of proposed GML-JO algorithm are shown in Algorithm 1.

\begin{remark}
Different from the Model-Agnostic Meta-Learning algorithm, the sub-networks of our proposed GML-JO are used to compute the solutions of sub-problems, and their parameters are updated according to the average loss of multiple sub-tasks.
\end{remark}
\begin{remark}
Different channel coefficients are used to learn, and the obtained solutions of two sub-problem are more stable comparing to alternating optimizatio.
\end{remark}

\begin{table}[tp]
\centering
\caption{NUMBER OF NEURONS IN THE NNs}
\begin{tabular}{|c|l|c|c|}
\hline
\textbf{No.} & \textbf{Layer Name}       & \textbf{Beamforming-Network}& \textbf{Antenna-Network}\\ \hline
1             & Input Layer               & $2 \times N$ $ \times M$& $N$                    \\ \hline
2             & Linear Layer 1            & 256& 256\\ \hline
3             & ReLU Layer                & 256& 256\\ \hline
4             & Output Layer              & $2 \times N$ $ \times M$& $N$                    \\\hline
\end{tabular}
\end{table}

\begin{algorithm}[tp]
\caption{Proposed GML-JO Algorithm}
\begin{algorithmic}[1]
\Procedure{GML-JO}{$\mathbf{P}, \mathbf{d}$}
    \State Randomly initialize $\theta_P, \theta_d, \mathbf{P}^{(0,1)}, \mathbf{d}^{(0,1)}$,$\gamma_m, y_m$
    \State $\mathbf{P}^* \gets \mathbf{P}^{(0,1)}$
    \State $\mathbf{d}^* \gets \mathbf{d}^{(0,1)}$
    \State Initialize the maximum WSR $\text{MAX} \gets 0$

    \For{$k \gets 1$ \textbf{to} $N_k$}
        \State $\overline{\mathcal{L}} \gets 0$
        \For{$j \gets 1$ \textbf{to} $N_j$}
            \State $\mathbf{P}^{(0,j)} \gets \mathbf{P}^{(0,1)}$
            \State $\mathbf{d}^{(0,j)} \gets \mathbf{d}^{(0,1)}$
                        \State Update $\gamma_m^{(j)}, y_m^{(j)}$
            \For{$i \gets 1$ \textbf{to} $N_i$}
                \State $\Delta \mathbf{P}^{(i-1,j)} \gets \text{BN}_{\theta_P}(\mathbf{g_p})$
                \State $\mathbf{P}^{(i,j)} \gets \mathbf{P}^{(i-1,j)} + \lambda \cdot \Delta \mathbf{P}^{(i-1,j)}$
                \State $\mathbf{P}^{(i,j)} \gets \text{normalize\_P}(\mathbf{P}^{(i,j)}, P_{\text{total}})$
            \State \hspace{-2\algorithmicindent}
            \EndFor \textbf{end for}
            \State $\mathbf{P}^* \gets \mathbf{P}^{(N_i,j)}$
            \For{$i \gets 1$ \textbf{to} $N_i$}
                \State $\Delta \mathbf{d}^{(i-1,j)} \gets \text{AN}_{\theta_d}(\nabla_{\mathbf{d}} f)$
                \State $\mathbf{d}^{(i,j)} \gets \mathbf{d}^{(i-1,j)} + \lambda \cdot \Delta \mathbf{d}^{(i-1,j)}$
                \State $\mathbf{d}^{(i,j)} \gets \text{Proj}_D(\mathbf{d}^{(i,j)})$
            \State \hspace{-2\algorithmicindent}
            \EndFor\textbf{end for}
            \State $\mathbf{d}^* \gets \mathbf{d}^{(N_i,j)}$
            \State $\mathcal{L}^j \gets -\text{WSR}(\mathbf{P}^*, \mathbf{d}^*)$
            \If{$-\mathcal{L}^j > \text{MAX}$}
                \State $\text{MAX} \gets -\mathcal{L}^j$
                \State $\mathbf{P}_{\text{opt}} \gets \mathbf{P}^*$
                \State $\mathbf{d}_{\text{opt}} \gets \mathbf{d}^*$
            \State \hspace{-2\algorithmicindent}
            \EndIf\textbf{end if}
            \State $\overline{\mathcal{L}} \gets \overline{\mathcal{L}} + \mathcal{L}^j$
        \State \hspace{-\algorithmicindent}
        \EndFor   \textbf{end for}
        \State $\overline{\mathcal{L}} \gets \frac{1}{N_j} \overline{\mathcal{L}}$
        \State $\theta_P \gets \text{Adam}(\theta_P, \nabla \overline{\mathcal{L}})$
        \State $\theta_d \gets \text{Adam}(\theta_d, \nabla \overline{\mathcal{L}})$
    \State \textbf{end for}
    \EndFor
    \State \Return $\mathbf{P}_{\text{opt}}, \mathbf{d}_{\text{opt}}$
    \State \textbf{end procedure}
\EndProcedure
\end{algorithmic}
\end{algorithm}



\section{Simulation Results}
\vspace{1.0em}
In this section, we evaluate the performance of the proposed GML-JO algorithm in multi-user-multi-waveguide  pinching-antenna communication systems through simulation experiments.

\subsection{Parameter Setting and Baseline}

The users are randomly distributed in a two-dimensional area of 20 m $\times$ 20 m, and the waveguides are positioned at equal distances within the area. 
The simulation parameters include -40dBm noise power, 3-meter-high waveguides, 28 GHz carrier frequency, 10 GHz cutoff wavelength, 1.4 effective refractive index, and a minimum SINR threshold is used to ensure the basic quality of service. The optimization algorithm adopts learning rates of \( \eta_D = 5 \times 10^{-4} \) for antenna position optimization and \( \eta_p = 2 \times 10^{-4} \) for beamforming vector optimization, and the number of iterations is set to be 500 in sub-task computation, and 50 channels are considered.

We compare the proposed GML-JO algorithm with the following baselines:

\begin{itemize}

\item \textbf{Gradient Meta-learning Optimization (GML):}
Unlike our proposed GML-JO, which uses equivalent substitution and convex approximation, GML directly optimizes the original problem \eqref{mosd2}. The SINR constraints are incorporated into the objective function as penalty terms, i.e.,
\begin{equation}
{E}(\mathbf{p}_m, \mathbf{d}) = \sum_{m=1}^M w_m \log_2 (1 + \text{SINR}_m) + \mu \sum_{m=1}^M V_m,
\end{equation}
where
\begin{equation}
V_m = \max(0, \gamma_m^{\text{min}} (I_m + \sigma^2) - G_m).
\end{equation}
GML employs two neural networks for optimization, where BN processes the gradient of the objective function $\nabla_{\mathbf{p}_m} E$ to generate updates for $\mathbf{p}_m$, enforcing power constraints through output normalization; while the AN utilizes the gradient $\nabla_{\mathbf{d}} E$ to produce updates for $\mathbf{d}$, confining it within the predefined region $S$ via a projection operator. The GML adopts an AO strategy, iteratively refining $\mathbf{p}_m$ and $\mathbf{d}$ to solve the original problem.

    \item \textbf{Equivalent Transformed Convex Approximation (ET-CA):}
Based on the lagrange duality and quadratic transformations, we introduce auxiliary variables \( \gamma_m \) and \( y_m \) to decouple variables, transforming the original non-convex WSR problem \eqref{mosd2} into an equivalent form  \eqref{mosd5}, and decomposes it into two subproblems: one for the beamforming vectors \( \mathbf{p}_m \) and one for the antenna position vector \( \mathbf{d} \). For the non-convex terms in the beamforming subproblem, convex approximation techniques are employed to transform it into a convex optimization problem. The SINR constraints are handled through the penalty function method. Unlike GML-JO, CVX is used instead of neural networks to solve the \(\mathbf{p}_m\) subproblems. The antenna position \(\mathbf{d}\) is updated via the projected gradient. Refer to Appendix A for further details.
    \item \textbf{AO:}
This method directly employs the conventional non-convex optimization problem through AO, circumventing the need for convex approximations, alternately updating the beamforming vectors \(\mathbf{p}_m\) and the antenna position \(\mathbf{d}\) through the gradient projection method. The SINR constraints are directly embedded into the objective function via the penalty function method and are gradually satisfied during the iteration process.
    \item \textbf{Gradient Descent (GD): }
The considered GD method optimizes the objective function by simultaneously updating the beamforming vectors \(\mathbf{p}_m\) and the antenna position vector \(\mathbf{d}\). Unlike AO, GD utilizes historical information during updates rather than fixing variable values, allowing flexible selection of iteration indices. The process begins by incorporating SINR constraints into the objective function as penalty terms. In the \( k \)-th iteration, the updates are given by
\begin{align}
\mathbf{p}^{(k+1)} &= \mathbf{p}^{(k)} - \eta \nabla_{\mathbf{p}} L(\mathbf{p}^{(k)}, \mathbf{d}^{(\tau_p)}), \label{eq:gd_p_update}
\end{align}
and
\begin{align}\mathbf{d}^{(k+1)} &= \mathbf{d}^{(k)} - \eta \nabla_{\mathbf{d}} L(\mathbf{p}^{(\tau_d)}, \mathbf{d}^{(k)}), \label{eq:gd_d_update}
\end{align}
where \(\tau_p \leq k\) and \(\tau_d \leq k\) are historical iteration indices, and \(\eta\) is the step size. After the updates, power constraints are enforced through normalization, and position constraints are implemented via projection.
    \item \textbf{Uniform Distribution Baseline (UDB): }
Instead of iterative optimization, \( \mathbf p_m \) and \( \mathbf d \) are directly set to uniformly distributed random values or averages, which are used as references for the lower bound on the system performance.
\end{itemize}
\begin{figure}[t]
  \centering
\includegraphics[width=0.5\textwidth]{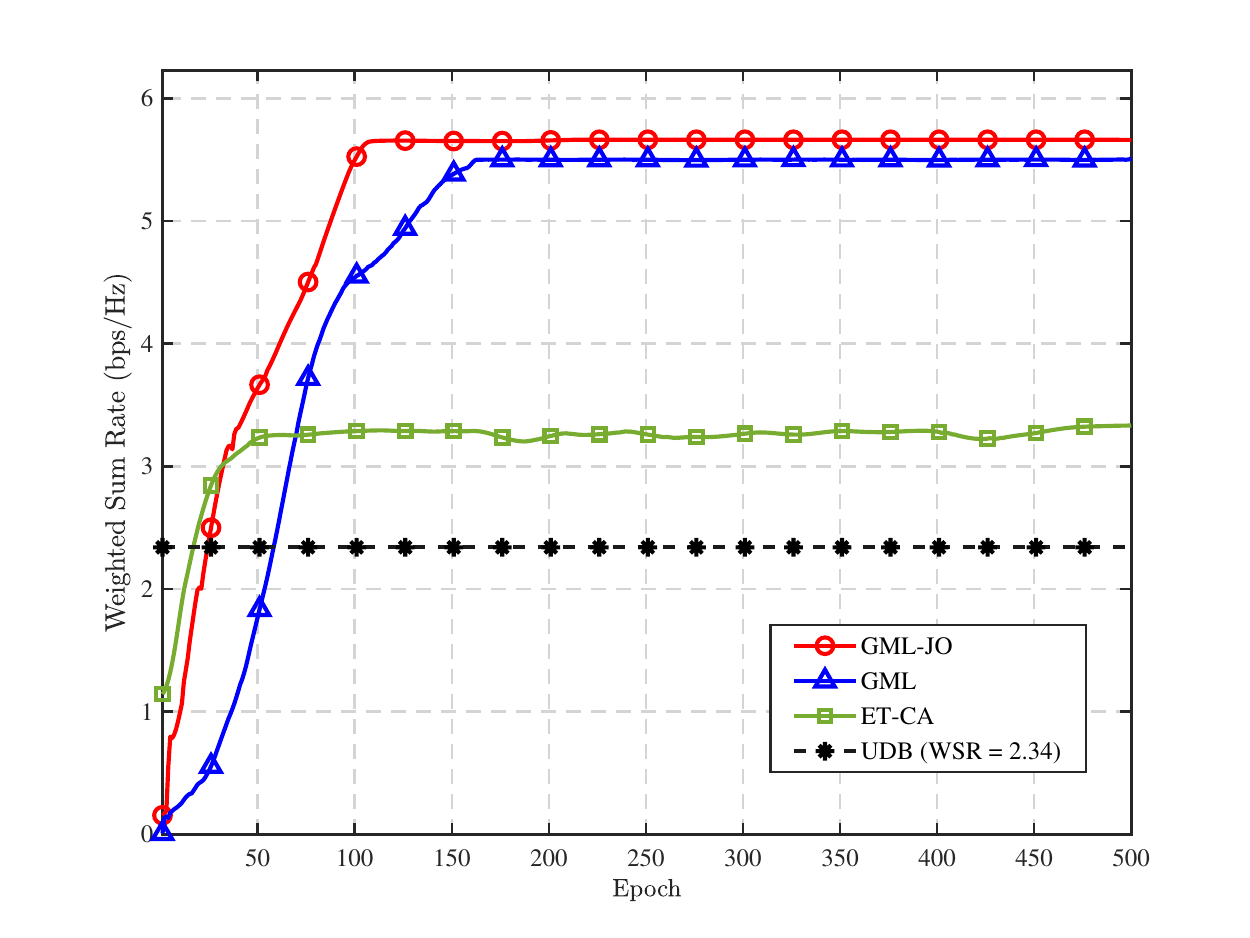}
  \caption{WSR performance comparison of different optimization algorithms.}
  \label{fig:example9876}
\end{figure}
\begin{figure}[t]
  \centering
\includegraphics[width=0.5\textwidth]{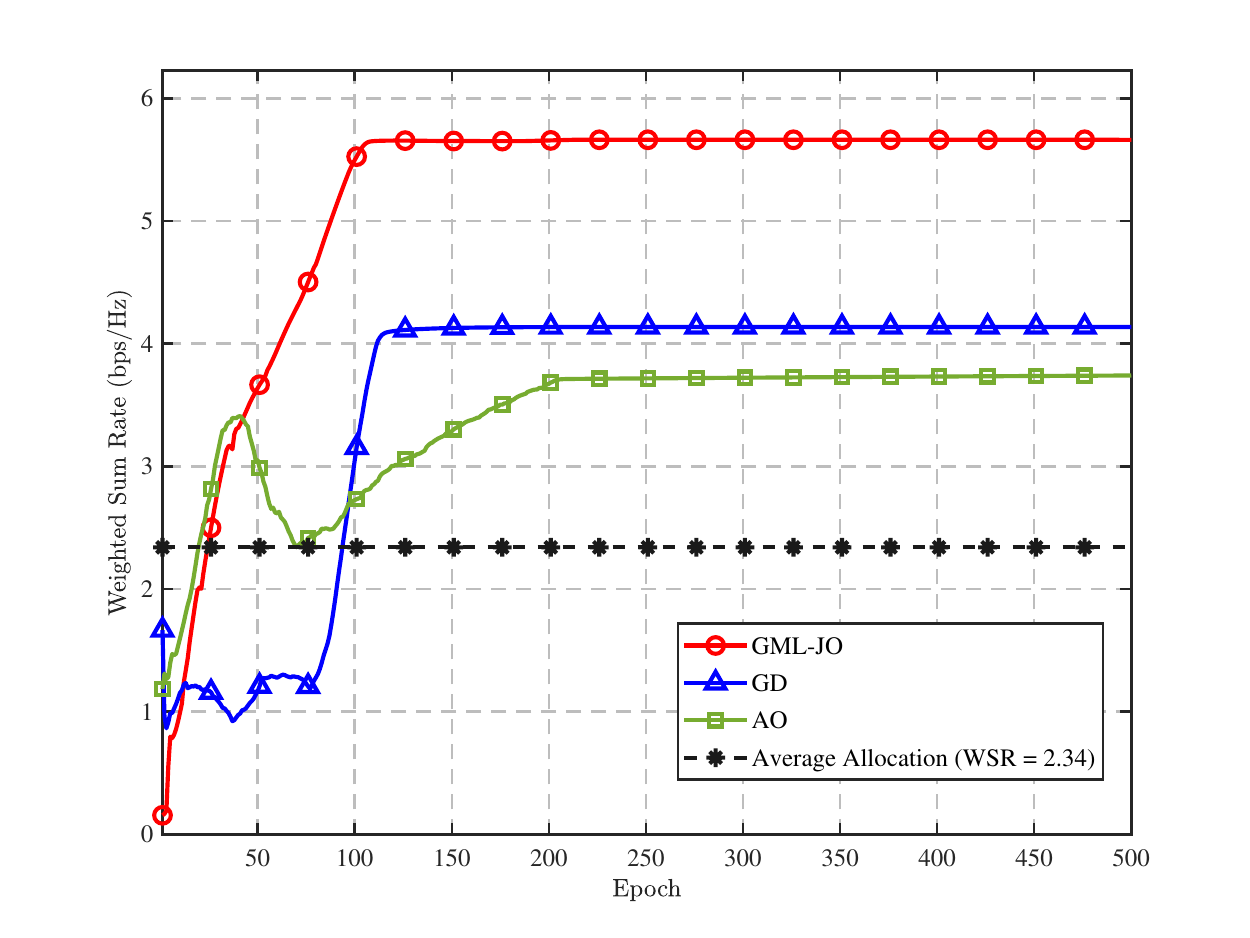}
  \caption{WSR performance comparison of GML-JO and traditional optimization algorithms.}
  \label{fig:example9875}
\end{figure}

We first evaluate the performance of the proposed GML-JO algorithm under fixed parameters. As illustrated in Fig. \ref{fig:example9876}, the WSR of our proposed GML-JO algorithm is 5.6 bits/s/Hz after 100 iterations, representing a 5.7\% improvement compared to the GML algorithm and a 24.4\% improvement compared to the ET-CA algorithm. Furthermore, our proposed GML-JO algorithm significantly outperforms the UDB in terms of WSR.

\begin{figure}[t]
  \centering
\includegraphics[width=0.5\textwidth]{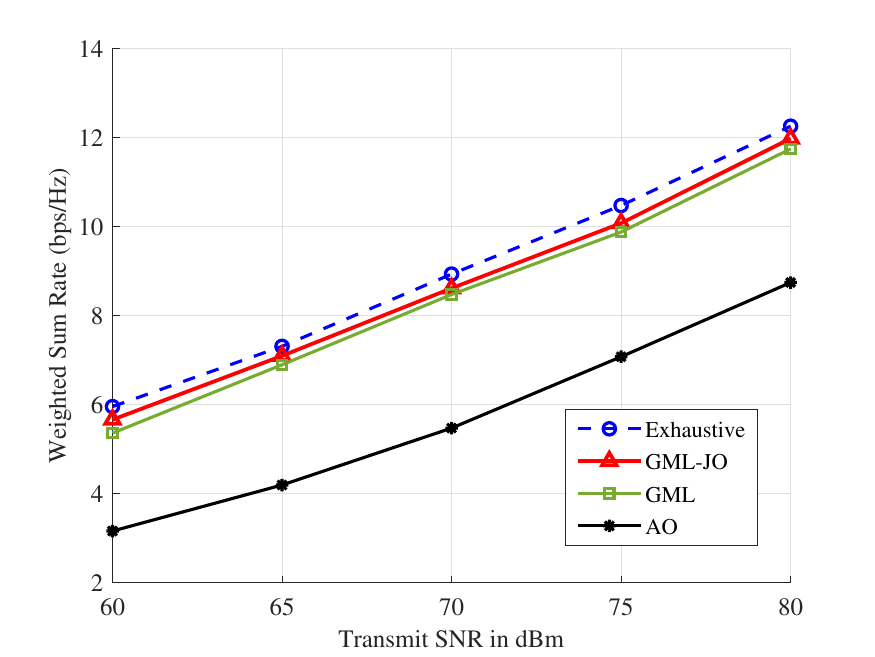}
  \caption{WSR performance comparison of different optimization algorithms.}
  \label{fig:example9874}
\end{figure}
\begin{figure}[t]
  \centering
\includegraphics[width=0.5\textwidth]{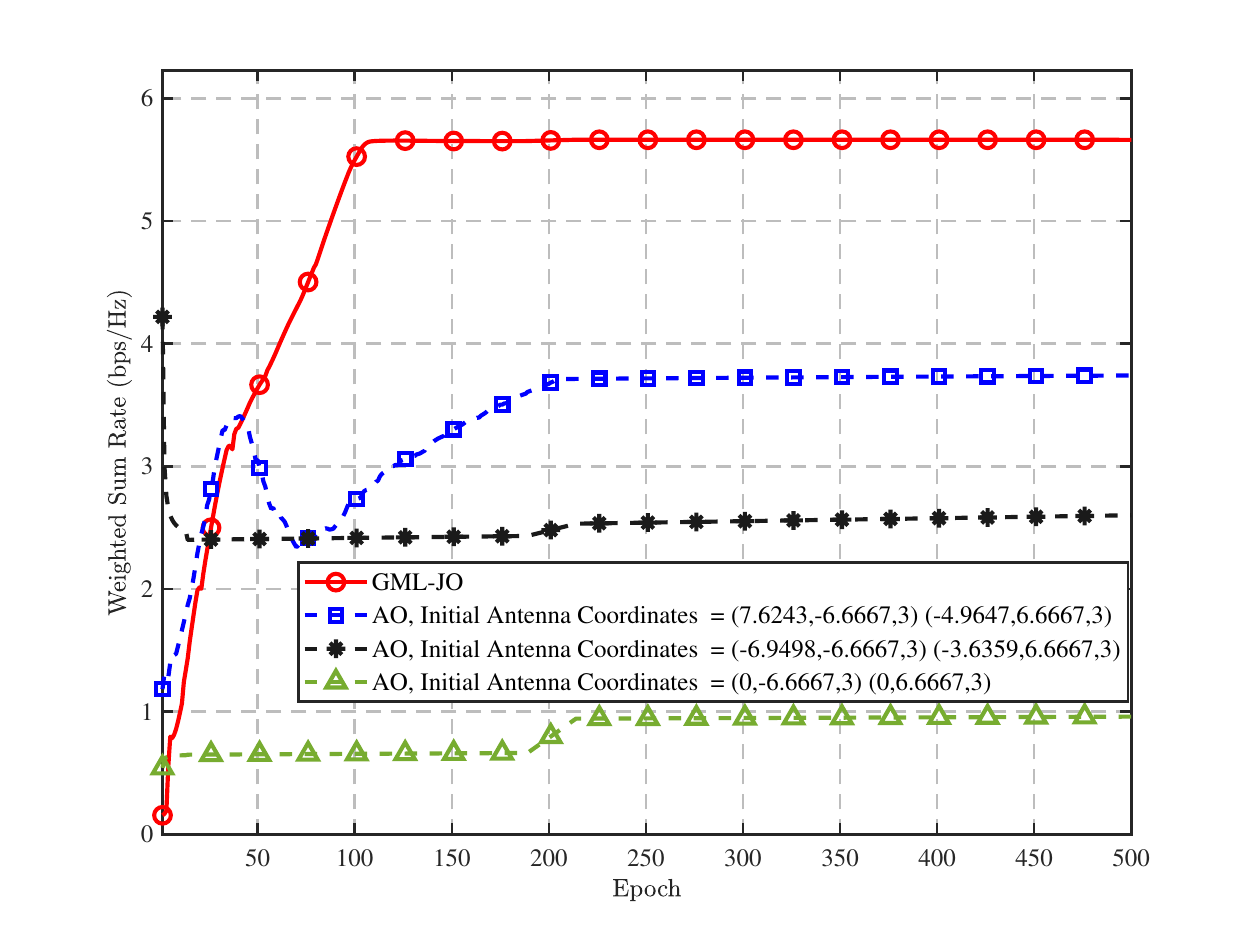}
  \caption{WSR performance comparison of GML-JO and AO under different initial antenna coordinates.}
  \label{fig:example9873}
\end{figure}
\subsection{ Algorithm Performance and Convergence Analysis}
\begin{figure}[t]
  \centering
\includegraphics[width=0.5\textwidth]{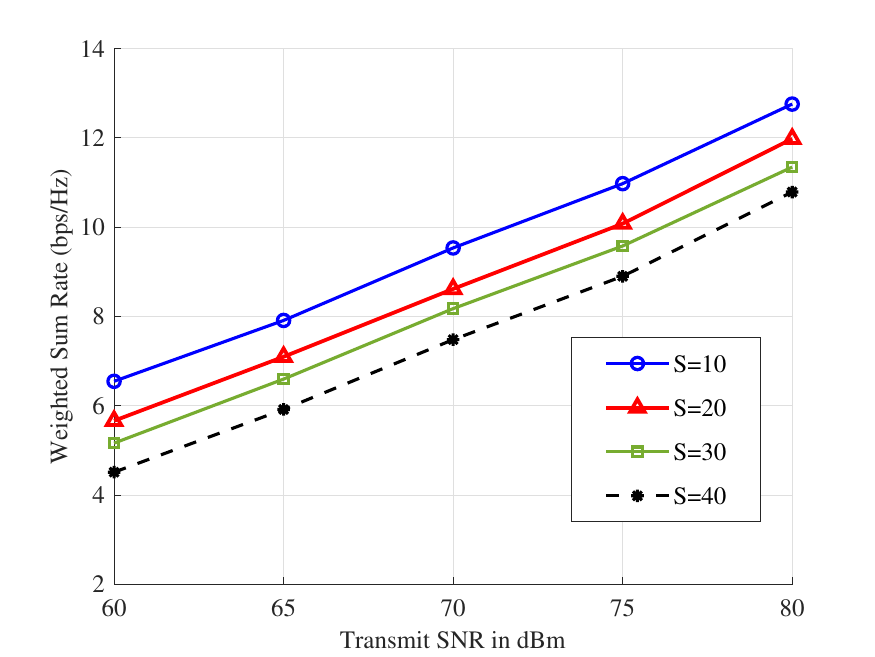}
  \caption{WSR performance of proposed GML-JO with different antenna positions.}
  \label{fig:example9872}
\end{figure}
\begin{figure}[t]
  \centering
\includegraphics[width=0.5\textwidth]{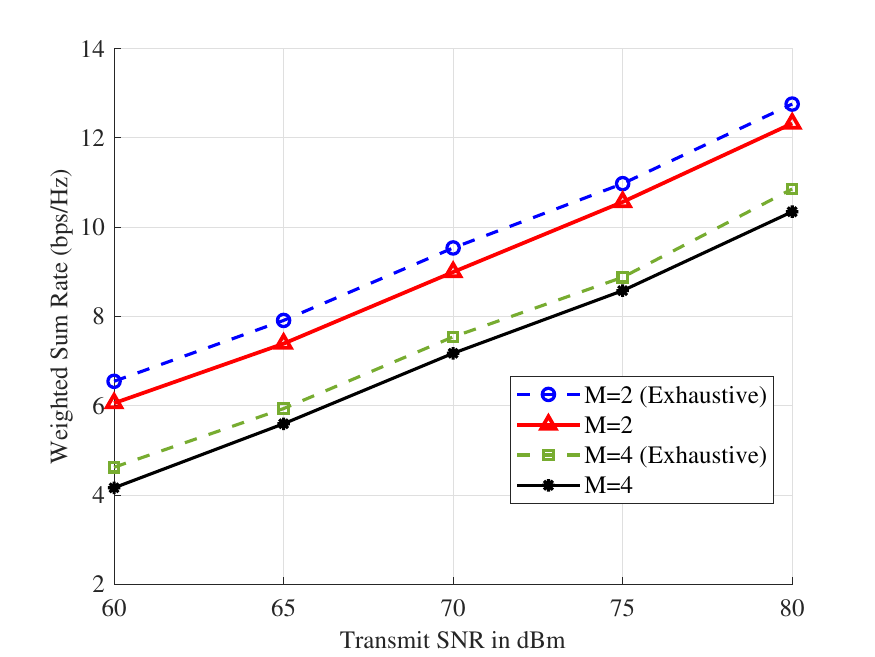}
  \caption{Performance comparison of  proposed GML-JO and exhaustive search with different numbers of users.}
  \label{fig:example9871}
\end{figure}
\begin{figure}[t]
  \centering
\includegraphics[width=0.5\textwidth]{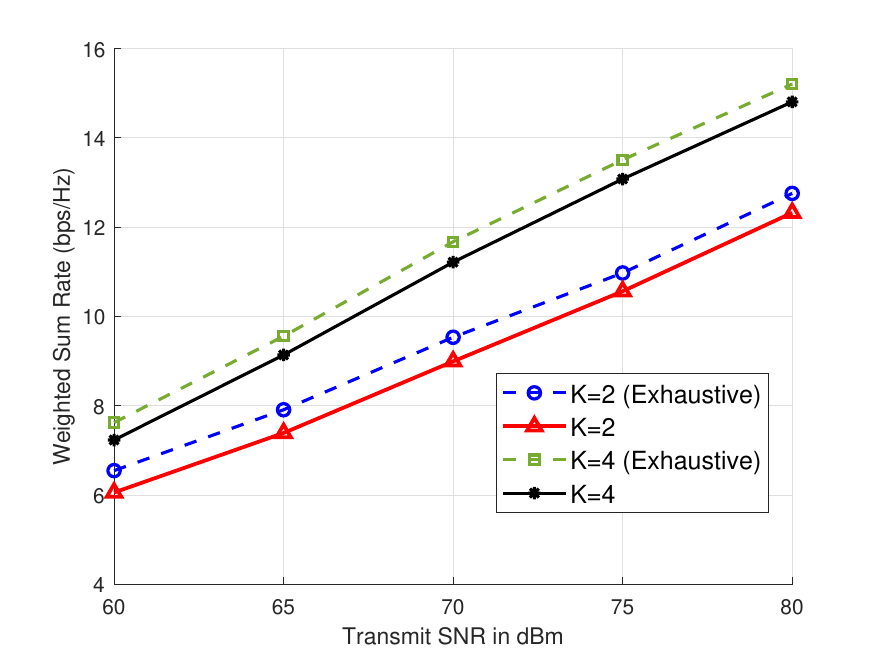}
  \caption{Performance comparison of proposed GML-JO and exhaustive search with different numbers of waveguides.}
  \label{fig:example9869}
\end{figure}
\begin{figure}[t]
  \centering
\includegraphics[width=0.5\textwidth]{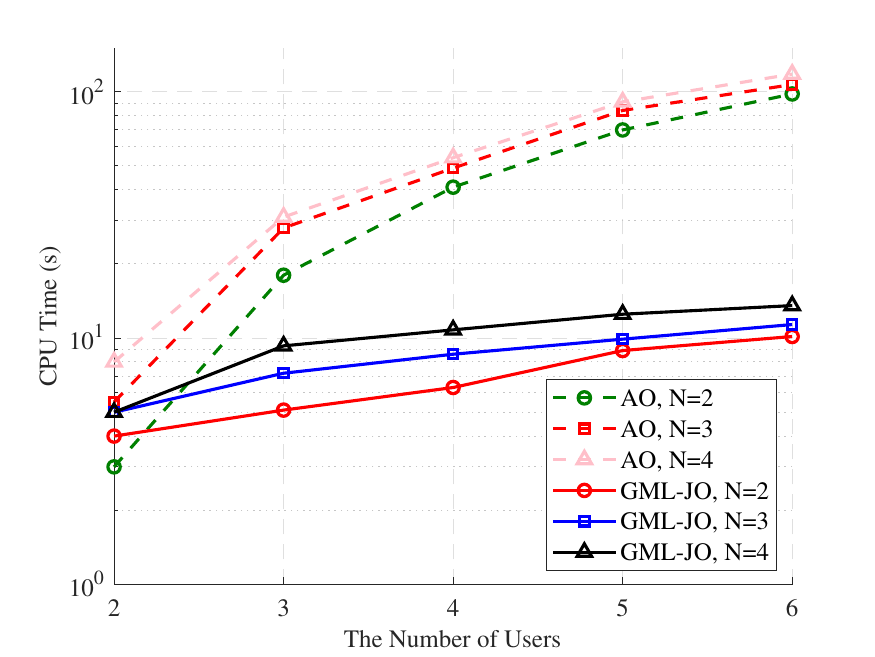}
  \caption{CPU execution time comparison of proposed GML-JO and AO.}
  \label{fig:example9870}
\end{figure}

Fig. \ref{fig:example9875} further evaluates the performance of our proposed GML-JO algorithm in a comparison to conventional optimization methods under an identical number of iterations. It demonstrates that the WSR of the proposed GML-JO algorithm is enhanced by 30.2\% and 16.7\% compared to the GD method and the AO method, respectively. Fig. \ref{fig:example9874} illustrates the impact of varying the transmit SNR from 60 dBm to 80 dBm on the WSR. As the SNR increases, the WSR of all methods exhibits a significant enhancement. In the SNR range $[60 \mathrm{dBm}, 70 \mathrm{dBm}]$, the performance differences among the algorithms remain minimal, with WSR increases ranging from approximately 2 to 4 bits/s/Hz. In contrast, within $[75 \mathrm{dBm},80 \mathrm{dBm}]$, the performance gap among the methods becomes more pronounced. Notably, the WSR achieved by the proposed GML-JO algorithm closely approaches the theoretical upper limit established by an exhaustive search, while maintaining substantially lower computational complexity.

\subsection{Parameter Impact Analysis}

Fig. \ref{fig:example9873} illustrates that the GML-JO algorithm exhibits high stability under different initial antenna coordinates. On the other hand, the performance of the AO method is sensitive to initialization, exhibiting WSR fluctuations of up to 15\%.

The effect of the antenna position adjustment range \( S \)  on the WSR is investigated in Fig. \ref{fig:example9872} for the tested values of 10, 20, 30, and 40 meters. For $S=20$, the WSR attains a peak of approximately 5.6 bits/s/Hz under high SNR conditions. As \( S \) increases to 40 meters, the WSR decreases to approximately 4.5 bits/s/Hz, with the reduction being particularly notable at low SNR, where it drops by about 40\%.

Fig. \ref{fig:example9871} evaluates the performance differences between the proposed GML-JO algorithm and the exhaustive method for scenarios with user counts of \( M = 2 \) and \( M = 4 \). For \( M = 2 \), the WSR of both methods are nearly identical, achieving approximately 5.6 bits/s/Hz at SNR = 60 $\mathrm{dBm}$. In contrast, when \( M = 4 \), the exhaustive method exhibits a slight performance edge at low SNR, reaching a WSR of about 4.4 bits/s/Hz at SNR = 60 $\mathrm{dBm}$; while the GML-JO algorithm sustains a robust performance with a WSR of approximately 4.1 bits/s/Hz at SNR = 60 $\mathrm{dBm}$. Fig. \ref{fig:example9869} evaluates the performance differences between the proposed GML-JO algorithm and the exhaustive method for scenarios with waveguide counts of $K=2$ and $K=4$, while fixing the number of users to $M=2$. For \( K = 2 \), the WSR of both methods are nearly identical, achieving approximately 5.6 bits/s/Hz at SNR = 60 $\mathrm{dBm}$. In contrast, when \( K = 4 \), the exhaustive method exhibits a slight performance edge at low SNR, reaching a WSR of about 7.8 bits/s/Hz at SNR = 60 $\mathrm{dBm}$; while the GML-JO algorithm sustains a robust performance with a WSR of approximately 7.6 bits/s/Hz at SNR = 60 $\mathrm{dBm}$.

As shown in Fig. \ref{fig:example9870}, we compare the average CPU execution time of GML-JO and AO across varying numbers of users $M$ and antennas $N$ with a fixed transmit power of 60 dBm. When \( M \) is fixed, the execution time of AO grows at about three times the rate of GML-JO as \( N \) increases. When \( N \) is fixed, the time of GML-JO grows slowly as \( M \) increases; while AO shows a sharp increase. Overall, the GML-JO has more than 10 times the speedup of the AO.

\section{Conclusion}
\vspace{1.0em}

In this paper, we considered a multi-waveguide multi-user pinching-antenna systems and formulated an optimization problem designed to maximize the WSR. The high-dimensional and non-convex nature of this problem poses significant challenges for conventional optimization methods. To address these issues, we proposed a GML-JO algorithm. Especially, proposed GML-JO employed an equivalent substitution technique to approximate and reframe the original problem, effectively decoupling the tightly intertwined dependencies between beamforming and antenna position optimization. Subsequently, we applied a convex approximation strategy to transform the non-convex problem into a tractable convex form, facilitating efficient optimization. Two dedicated sub-networks were subsequently employed to  address the sub-optimization tasks. Ultimately, we obtained the average rate across multiple sub-tasks computations, delivering a robust and effective solution. Experimental results demonstrate that the proposed GML-JO achieved rapid convergence to a WSR of 5.6 bits/s/Hz within 100 iterations, delivering 32.7\% superior performance enhancement compared to the benchmark AO method while maintaining significantly reduced computational complexity. Notably, the obtained solution exhibited remarkable robustness against variations in random user distributions, extensive antenna positioning adjustments, and dynamic channel conditions.

\section*{\textbf{Appendix A}}
\begin{center}
\textbf{ET-CA Baseline Algorithm}
\end{center}

Based on the Lagrange duality and quadratic transformations, we introduce auxiliary variables \( \gamma_m \) and \( y_m \) to decouple variables, transforming the original non-convex WSR problem \eqref{mosd2} into an equivalent form \eqref{mosd5}, and decompose it into two subproblems. The auxiliary variables update are given by \eqref{mosd520} and \eqref{mosd521}, and the beamforming coefficient and antenna position updates of ET-CA are given as follows:

\subsubsection{Beamforming Coefficient $\mathbf{p}_m$ Update}
For given $y_m^{(t)}$, $\mathbf{d}^{(t)}$, and $\gamma_m^{(t)}$, the sub-problem of problem \eqref{mosd5} on variable $\mathbf{p}_m$ can be formulated as \eqref{sldp32}. Since $G_m$ and $I_m$ are nonconvex functions with respect to variable $\mathbf{p}_m$, the objective function and the $\mathrm{SINR}_m$ constraint in subproblem \eqref{sldp32} are nonconvex. To solve this problem, we perform a first-order Taylor expansion of $G_m$ and $I_m$ at $\mathbf{p}_m^{(t)}$ by \eqref{mosd522} and \eqref{mosd523}.

Then, the subproblem \eqref{sldp32} can be approximated by
\begin{subequations}\label{lotr1}
\begin{align}
& \max_{\mathbf{p}_m} \tilde{f}_p(\mathbf{p}_m ) \\
&\text{s.t.} \quad \eqref{sldo611}, \eqref{sinre2}.
\end{align}
\end{subequations}

Using the quadratic penalty function method, problem \eqref{lotr1} only with constraint \eqref{sinre2} is transformed into an unconstrained problem, i.e.,

\begin{equation}\label{oier5}
\max_{\mathbf{p}_m } \tilde{f}_p^{\text{new}}(\mathbf{p}_m ) = \tilde{f}_p(\mathbf{p}_m ) - \mu \sum_{m=1}^{M} V_m^2(\mathbf{p}_m ).
\end{equation}
To handle power constraint \eqref{sldo611}, we introduce an obstacle function, defined as
\begin{equation}\label{lorte}
\varphi(\mathbf{p}_m) = - \log \left( P_{\text{total}} - \sum_{m=1}^{M} \| \mathbf{p}_m \|^2 \right).
\end{equation}

By considering the constraint \eqref{sldo611}, the objective function of unconstrained problem \eqref{oier5} is further expressed as
\begin{equation}
Z(\mathbf{p}_m) = \tilde{f}_{p}^{\text{new}}(\mathbf{p}_m) + \frac{1}{t} \varphi(\mathbf{p}_m),
\end{equation}
where $t > 0$ is the obstacle parameter.

The interior point method maximises $Z(\{\mathbf{p}_m\})$ using gradient ascent method, where the gradient  calculation is given by
\begin{equation}\label{lopt32}
\nabla_{\mathbf{p}_m} Z = \nabla_{\mathbf{p}_m} \tilde{f}_p^{\text{new}} + \frac{1}{t} \nabla_{\mathbf{p}_m} \varphi(\mathbf{p}_m).
\end{equation}
From \eqref{oier5} and \eqref{lorte}, we have
\begin{equation}\label{rty6}
\nabla_{\mathbf{p}_m} \tilde{f}_p^{\text{new}} = \nabla_{\mathbf{p}_m} \tilde{f}_p - 2\mu \sum_{m=1}^{M} V_m \nabla_{\mathbf{p}_m} V_m,
\end{equation}
and
\begin{equation}\label{port7}
\nabla_{\mathbf{p}_m} \varphi(\mathbf{p}_m) = \frac{2 \mathbf{p}_m}{P_{\text{total}} - \sum_{m=1}^{M} \| \mathbf{p}_m \|^2}.
\end{equation}

Submitting \eqref{rty6} and \eqref{port7} into \eqref{lopt32}, the update of $\{\mathbf{p}_m\}$ is expressed as
\begin{equation}
\mathbf{p}_m^{(k+1)} = \mathbf{p}_m^{(k)} + \alpha \nabla_{\mathbf{p}_m} Z.
\end{equation}

\subsubsection{Antenna Position $\mathbf{d}$ Update}
With fixing $y_m^{(t)}$, $\mathbf{p}_m^{(t+1)}$, and $\gamma_m^{(t)}$, the subproblem of problem \eqref{mosd5} on variable $\mathbf{d}$ can be formulated as \eqref{lotr5}.

With the calculation of the gradient $\nabla_{\mathbf{d}} f(\mathbf{d})$, the update of $\mathbf{d}$ is given by

\begin{equation}
\mathbf{d}^{(t+1)} = \mathbf{d}^{(t)} + \alpha \nabla_{\mathbf{d}} f \left( \mathbf{d}^{(t)} \right),
\end{equation}
where $\alpha$ is the step size. To ensure that $\mathbf{d}$ satisfies the constraint \eqref{lert3}, the updated $\mathbf{d}^{(t+1)}$ needs to be projected to the feasible set $D$, i.e.,

\begin{equation}\label{rty56}
\mathbf{d}^{(t+1)} = \text{proj}_D \left( \mathbf{d}^{(t)} + \alpha \nabla_{\mathbf{d}} f \left( \mathbf{d}^{(t)} \right) \right).
\end{equation}

\vspace{0.9em}
\bibliographystyle{IEEEtran}
\bibliography{IEEEfull,trans}
\vspace{0.9em}

\end{document}